\documentclass{article}

\usepackage{amsmath,amssymb,upgreek,graphicx,theorem,fancyhdr,natbib}

\newtheorem{theorem}{Theorem}[section]

\newenvironment{proof}{\textbf{Proof.}}{\hfill$\Box$}

\newcommand{\x}{\mathbf{x}}
\newcommand{\dd}{\mathrm{d}}
\newcommand{\D}{\mathrm{D}}
\newcommand{\ii}{\mathrm{i}}
\newcommand{\e}{\mathrm{e}}
\newcommand{\re}{\mathrm{Re}\,}
\newcommand{\im}{\mathrm{Im}\,}
\newcommand{\OO}{\mathrm{O}}
\newcommand{\upi}{\uppi}

\author{Andrzej Hanyga\\
ul. Bitwy Warszawskiej 1920r. 14/52\\
Warszawa, PL\\
{\tt ajhbergen@yahoo.com} 
}

\title{Attenuation and shock waves in linear hereditary viscoelastic media. Strick-Mainardi 
and Jeffreys-Lomnitz-Strick creep compliances.} 

\begin{document}

\maketitle


\begin{abstract}
Dispersion, attenuation and wavefronts in a class of linear viscoelastic media proposed by 
Strick and Mainardi in 1982 and a related class of models due to Lomnitz, Jeffreys and Strick are 
studied by a new method due to the Author. Unlike the 
previously studied explicit models of relaxation modulus or creep compliance, these two 
classes support propagation of discontinuities. Due to an extension made by Strick either 
of these two classes of models comprise both viscoelastic solids and fluids. 
\end{abstract}

\section{Introduction.}

In most explicit analytic models of viscoelastic media the attenuation as a function of frequency tends to infinity
according to a power law. As a consequence in such viscoelastic media initial discontinuities and discontinuous 
source signals are immediately smoothed out. In those media in which additionally disturbances are bounded in space
by a wavefront propagating at a finite speed the wavefield must decay to zero with all its derivatives 
at the wavefront. Consequently the peak of a pulse stays behind the wavefront and is preceded by a flat 
pedestal \citep{Strick1:ConstQ}. The pedestal widens with the propagation and the delay of the signal with respect
to the wavefront increases in time \citep{HanQAM,HanSerWM}. The delay of seismic signals with respect to the 
wavefronts has to be taken into account in seismic inversion \citep{Strick3:ConstQ,FastTrack}. 

Viscoelastic models with a power law behavior in the high frequency limit are common in
materials science (e.g. in polymer and rubber theory) and in the theory of ultrasound in biotissues 
\cite{SzaboWu00,Szabo}. Biot's theory of poroelasticity \citep{BiotMech:BIOT,BiotMechP:BIOT,BiotI,BiotII} leads 
to similar results \citep{HanCarc,LuHanygaPorous3,LuHanygaJCP,LuHanygaGP}.

In this paper we shall present a class of creep compliances proposed by seismologists Lomnitz, Jeffreys, Mainardi and Strick \citep{Lomnitz62,Jeffreys67,StrickMainardi82,Strick84}.
These creep compliances were originally considered in connection with the constant $Q$ hypothesis. They however 
deserve attention because of another peculiarity: in the viscoelastic media defined by these creep compliances the attenuation
function is bounded and therefore such media support propagation of discontinuities at the wavefront.
Green's function for such media can be locally decomposed into a discontinuity wave and a continuous remainder
\begin{equation}
\mathcal{G}(t,\x) = a(\x)\, \theta(t - F(\x)) + \mathcal{G}_1(t,\x)
\end{equation}
If a pulse $f^\prime(t)\, \delta(\x)$ is sent from a point source then the wavefield
\begin{equation} \label{eq:conv} 
u(t,\x) = f^\prime(t)\ast \mathcal{G}(t,\x) = a(\x)\, f(t - F(\x)) + f^\prime(t)\ast\mathcal{G}_1(t,\x)
\end{equation}
where 
\begin{equation}
\varphi_1(t)\ast\varphi_2(t) := \int_0^t \varphi_1(s)\,\varphi_1(t-s)\, \dd s
\end{equation} 
denotes the Volterra convolution with respect to time. 
Equation~\eqref{eq:conv} shows that the pulse travels with the speed of the wavefront. 
This is an assumption 
commonly made in seismic inversion. It is clear that a 
careful analysis of viscoelastic models of wave propagation is overdue. 

Wavefronts in Jeffreys media ($\alpha > 0$) 
were previously studied numerically by \citet{Buchen1:VE}, who summed ray expansions and compared 
the wavefronts for various pulse shapes and material parameters. Our objective is to 
put wavefront discontinuities and the attenuation functions in the same perspective. 
Low-frequency attenuation has often been studied by experimental methods  
in materials science, bio-tissues and in seismology. Wavefront singularities provide additional information
on attenuation in the high-frequency range. Wavefronts singularities are also relevant for a correct 
definition of travel time \citep{FastTrack}. 
However, since attenuation is often considered independently
of wavefront singularities and pulse propagation, viscoelastic models for these two kinds 
of phenomena are often inconsistent.

It was shown in \citet{SerHanJMP} and \citet{HanWM2013} that the propagation speed 
$c(\omega)$ and the attenuation function $\mathcal{A}(\omega)$ in a viscoelastic medium 
with a creep compliance which is a Bernstein function can be expressed in terms of a Radon measure (essentially a locally finite measure) called 
the attenuation spectrum. Only the low-frequency behavior of the propagation speed and the attenuation function 
is available to experiments \citep{HanWM2013,NasholmHolm2011}. 
In \citet{HanJCA} it was shown that the high-frequency asymptotics of the attenuation function
determines the regularity of viscoelastic Green's functions. In \citet{SerHanJMP} and \citet{HanWM2013} a 
causal function $g(t)$ was defined such that $\mathcal{A}(\omega) = \re [p\, \tilde{g}(p)]$. 
In \citet{HanUno} asymptotic estimates and upper bounds  of the Green's functions near 
the wavefront have been expressed in terms of the function $g(t)$. 

We shall apply this analytic toolbox to the analysis of attenuation, dispersion and 
discontinuity waves in two classes of viscoelastic models: Strick-Mainardi models and 
the Jeffreys-Lomnitz-Strick models. The attenuation function and the function $g(t)$ can be 
explicitly calculated for Strick-Mainardi models. Both classes comprise viscoelastic solids 
($\alpha < 0$) and viscoelastic fluids ($\alpha \geq 0$). 

\section{Mathematical preliminaries.}

We shall consider the Initial-Value Problem (IVP) 
\begin{gather} \label{eq:mDlinpro}
\rho \, u_{,tt} = \nabla\cdot[G(t)\ast \nabla u_{,t}] + \delta(x)\,\delta(t), 
\qquad t \geq 0,\quad x \in \mathbb{R}\\
u(0,x) = 0; \quad u_{,t}(0,x) = 0 \label{eq:IVmDlinpro}
\end{gather} 
for the particle velocity $u$ in a hereditary viscoelastic medium. It is assumed that 
the relaxation modulus $G(t)$ (defined for $t > 0$) is completely monotonic (CM), i.e. it has derivatives $\D^n\, G$ of
arbitrary order and these derivatives satisfy the inequalities 
$$(-1)^n \, \D^n\, G(t) \geq 0 \qquad\text{on $\mathbb{R}$ for $n =0,1,2\ldots$}$$
It is also assumed that $G$ is locally integrable, or, equivalently
$$\int_0^1 G(s)\, \dd s < \infty$$
We shall use the abbreviation LICM for locally integrable completely monotonic functions. 
It follows \citep{HanDuality} that the creep compliance $J(t)$ ($t \geq 0$), related to the relaxation 
modulus by the equation
\begin{equation} \label{eq:duality}
\int_0^t G(s)\, J(t - s)\, \dd s = t \qquad\text{for $t \geq 0$}
\end{equation}
is a Bernstein function (BF), i.e. it is non-negative, differentiable and its derivative $J^\prime$ is LICM 
\citep{BernsteinFunctions}. Conversely, for a given BF $J$ equation~\eqref{eq:duality} has a unique solution $G$ 
and the solution $G$ is LICM \cite{HanDuality}.
We also recall that $0 \leq J_0 := J(0) < \infty$ and $J_0 = 0$ if and only if 
$\lim_{t\rightarrow 0+} G(t) = \infty$.

The solution of the IVP (\ref{eq:mDlinpro}--\ref{eq:IVmDlinpro}) is given by the formula
\begin{equation} \label{eq:Green1D}
u(t,x) = \frac{1}{4 \upi \ii} \int_{-\ii \infty + \varepsilon}^{\ii \infty + \varepsilon}
\frac{\kappa(p)}{2 \rho\, p^2}\, \e^{p\, t -\kappa(p)\, \vert x \vert} \, \dd p
\end{equation}
where \begin{equation} \label{eq:kappadef}
\kappa(p) := \rho^{1/2}\, p \, \left[p \, \tilde{J}(p)\right]^{1/2}
\end{equation}
and $\varepsilon > 0$.

In \citet{SerHanJMP} and \citet{HanWM2013} it was showed that $\kappa(p)$ is a complete 
Bernstein function (CBF) \citep{BernsteinFunctions,Jacob01I}, i.e.
$$\kappa(p) = p^2\, \tilde{F}(p),$$ where $F$ is a Bernstein function.
Furthermore $\kappa(0) = 0$.
Consequently $\kappa$ has an integral representation of the following form
\begin{equation} \label{eq:kappa}
\kappa(p) = p/c_0 + p \int_{]0,\infty[} \frac{\nu(\dd r)}{p + r}
\end{equation}
where $\nu$ is a positive Radon measure satisfying the inequality
\begin{equation} \label{eq:doss}
\int_{]0,\infty[} \frac{\nu(\dd r)}{1 + r} < \infty 
\end{equation}
\citep{BernsteinFunctions} and $c_0$ is a constant satisfying the inequalities 
$0 < c_0 \leq \infty$, defined by the formula 
\begin{equation}
1/c_0 := \lim_{p\rightarrow\infty} \kappa(p)/p 
\end{equation}
Note that
\begin{equation} \label{eq:c0}
1/c_0 = \rho^{1/2}\,\lim_{p\rightarrow \infty} \left[p\, \tilde{J}(p)\right]^{1/2} = 
[\rho\, J_0]^{1/2}
\end{equation}
The dimension of $\kappa(p)$ and $\nu(\dd r)$ is 1/L. We shall assume that $J_0 > 0$ and $c_0 < \infty$. 
This excludes some viscoelastic models used in seismology in connection with the constant $Q$
hypothesis (e.g. \citet{Kjartansson:ConstQ}) and in materials science in connection with the power law 
attenuation (e.g. \citet{KellyMcGoughMeerschaert08}). 
 
If $J_0 > 0$ then the constant $c_0$ defines the wavefronts $\vert x\vert = c_0\, t$ such that 
$u(t,x) = 0$ for $t > \vert x \vert/c_0$, otherwise $c_0 = \infty$ and the solution 
$u(t,x)$ does not vanish anywhere in the space-time. 

The attenuation function $\re \kappa(-\ii \omega)$ and the dispersion function 
$-\im \kappa(-\ii \omega)$ of the medium  
can be expressed in terms of the Radon measure $\nu$, hence the Radon measure  $\nu$ 
is called the dispersion-attenuation measure in \cite{HanWM2013}
\begin{gather}
\mathcal{A}(\omega) = \omega^2 \int_{]0,\infty[} \frac{\nu(\dd r)}{r^2 + \omega^2} \label{eq:attn}\\
\mathcal{D}(\omega) = \omega \int_{]0,\infty[} \frac{r\, \nu(\dd r)}{r^2 + \omega^2}
\end{gather}

The attenuation function $\mathcal{A}(\omega)$ is non-decreasing and therefore it tends
to a finite limit $\mathcal{A}_\infty :=
\lim_{\omega\rightarrow\infty} \mathcal{A}(\omega)$ if it is bounded.
If $\nu$ has finite mass $N := \nu(]0,\infty[)  < \infty$ then 
$\lim_{\omega\rightarrow\infty} \mathcal{A}(\omega) = N$ by the Lebesgue Dominated Convergence Theorem. 
In particular $N < \infty$ if the support of $\nu$ is bounded. Conversely, if 
$\mathcal{A}(\omega)$ is bounded, then, 
by the Fatou lemma (\cite{Rudin76}, Theorem~11.31) and equation~\eqref{eq:attn} 
$\int_0^\infty \nu(\dd r) \leq \lim_{\omega\rightarrow\infty} \,
\mathcal{A}(\omega)$ and the attenuation-dispersion spectral measure $\nu$ has finite mass. 
By the 
preceding argument $\mathcal{A}_\infty = N$. We have thus proved that $\mathcal{A}_\infty = N$ and
both numbers can be finite or infinite.

The following theorem can be used to check whether the attenuation measure $\nu$ has finite 
total mass.
\begin{theorem} \label{thm:xx}
\begin{equation} \label{eq:oo}
\lim_{p \rightarrow \infty} \left[p\, \left(\frac{\kappa(p)}{p} - \frac{1}{c_0}\right) \right]
= \int_{]0,\infty[} \nu(\dd r)
\end{equation}
where the right-hand side can be infinite.
\end{theorem}
\begin{proof}
The left-hand side of equation~\eqref{eq:oo} equals $ \int_{]0,\infty[} [p/(r + p)]\,\nu(\dd r)$. The theorem follows by the Monotone Convergence Theorem \citep[Sec.~11.28]{Rudin76}.\hfill
\end{proof}

In terms of the creep compliance
\begin{equation} \label{eq:yy}
\lim_{p\rightarrow \infty} \left\{ p \left[\left(\rho p \tilde{J}(p)\right)^{1/2} - (\rho\, J_0)^{1/2}
\right] \right\} = \int_{]0,\infty[} \nu(\dd r)
\end{equation}

The Radon measure $\nu$ can be calculated using equation~\eqref{eq:kappa}. 
If $\nu(\dd r) = h(r)\, \dd r$, then
\begin{equation} \label{eq:hrkappa}
h(r) = \frac{1}{\upi} \im \left[ \kappa(p)/p\right]_{p=r \, \exp(-\ii \upi)}
\end{equation}
\citep{SerHanJMP,HanWM2013}, or, using equation~\eqref{eq:kappadef},
\begin{equation} \label{eq:hrJ}
h(r) = \frac{\rho^{1/2}}{\upi} \im \left\{ \left[p \, \tilde{J}(p)\right]^{1/2}\right\}
\end{equation}

Recall that every LICM function $\varphi$ has the integral representation 
\begin{equation} \label{eq:LICM}
\varphi(t) = a + \int_{]0,\infty[} \e^{-r\, t} \,\nu(\dd r)
\end{equation}
where $\nu$ is a positive Radon measure satisfying the inequality \eqref{eq:doss}
\citep{GripenbergLondenStaffans}.
Define the function function $g$ by the formula
\begin{equation} \label{eq:KK2}
g(t) = \int_{]0,\infty[} \e^{-r\, t} \, \nu(\dd r)
\end{equation}
where the Radon measure $\nu$ is defined by equation~\eqref{eq:kappa}. 
We then have an important formula
\begin{equation}\label{eq:important}
\kappa(p) = \frac{p}{c_0} + p\, \tilde{g}(p)
\end{equation}
The function 
$g$ is LICM and $\lim_{t\rightarrow \infty} g(t) = 0$. The dimension of $g(t)$ is 1/L. 
The function $g(t)$ assumes a finite value at 0 if $\nu$ has a finite mass. Note that 
any function $\kappa$ given by equation~\eqref{eq:important}, where $g$ is a LICM function, 
is a CBF on account of equation~\eqref{eq:LICM} and equation~\eqref{eq:kappa}. Furthermore, 
it is proved in \citet{HanUno} that 
\begin{equation}\label{eq:kli}
g(0+) = \rho \, c_0 \, J^\prime(0+)/2.
\end{equation}
or, equivalently, 
\begin{equation} \label{eq:g0vsJ}
g(0+) = J^\prime(0+)/(2 J_0\, c_0)
\end{equation}
Furthermore
\begin{equation} \label{eq:ineqg}
g(t) \leq \rho \, c_0 \, J^\prime(t)/2
\end{equation}

If the attenuation function is bounded then $g(0+) = \int_{]0,\infty[} \nu(\dd r) = \mathcal{A}_\infty < \infty$.
 
Green's function $\mathcal{G}$ can be approximated by an explicit function $H(t,x)$
\begin{equation} \label{eq:GtoH}
\mathcal{G}(t,x) = \frac{1}{2 \rho} H(t - \vert x \vert/c_0, \vert x \vert) \, 
[1 + \OO[t - \vert x \vert/c_0]]
\end{equation}
where $H(\cdot,r)$ is a non-negative non-decreasing function defined by the equation 
\begin{equation}
\e^{-p\, \tilde{g}(p)\,r}/p = \int_0^\infty \e^{-p\, t} \, H(t,r)\, \dd t
\end{equation}
It is then proved in \citet{HanUno} that 
\begin{equation}
H(t,r) \sim_{t\rightarrow 0} \e^{-g(t)\, r}
\end{equation}
In view of equation~\eqref{eq:GtoH} this implies that 
\begin{equation}
\mathcal{G}(t,x) \sim_{t\rightarrow \vert x\vert/c_0 + 0} \frac{1}{2 \rho} \e^{-g(t - \vert x\vert/c_0)\, r}
\end{equation}

If $g(0+) < \infty$ then it is also true that
\begin{equation}
\lim_{t\rightarrow \vert x\vert/c_0 + 0} \mathcal{G}(t,x) = \frac{1}{2 \rho} \e^{-g(0+)\, r}
\end{equation}
while $\lim_{t\rightarrow \vert x\vert/c_0 - 0} \mathcal{G}(t,x) = 0$. Hence in this case
the wavefront carries a jump discontinuity $\exp(-g(0+) \,r)/(2 \rho)$.

\section{The Strick-Mainardi creep compliance.}

Consider the following function
\begin{equation}
F_\alpha(\Omega,p) := \frac{1}{\alpha} 
\left[\left(1 + \frac{\Omega}{p}\right)^\alpha - 1\right], \qquad -1 < \alpha < 1,\quad\alpha \neq 0
\end{equation}
and its limit for $\alpha \rightarrow 0$:
\begin{equation}
F_0(\Omega,p) = \ln\left( 1 + \frac{\Omega}{p}\right)
\end{equation}
We suppose that $F_\alpha$ is a Laplace transform and try to find its original $K_\alpha$:
$$K_\alpha(t,\Omega) = \frac{1}{2 \upi  \, \alpha\, \ii} \int_\mathcal{B} \e^{p\, t}\, 
\left[\left(1 + \frac{\Omega}{p}\right)^\alpha - 1\right]\, \dd p$$
It follows from the asymptotic estimate $(1 + \Omega/p)^\alpha - 1 \sim_\infty \alpha \, 
\Omega/p$ that 
the integrand tends to zero uniformly for $p \rightarrow \infty$ in the left complex half-plane 
$\re p \leq 0$.
The integrand on the right-hand side does not have any singularities outside the
cut along the negative real semi-axis. There is no contribution of the small circle of radius 
$\varepsilon$ centered at the origin. By Jordan's lemma the Bromwich contour can 
be replaced by the Hankel loop encircling the negative semi-axis in the positive direction. Setting $p = r\, \e^{\ii \upi}$
for the part of the contour running above the cut yields the following expression:
\begin{equation} \label{eq:Gxx}
K_\alpha(t,\Omega) = \frac{-1}{\alpha \upi } \int_0^\infty \e^{-r\, t} \, \im \left(1 + \e^{-\ii \upi}\,\frac{\Omega}{r} 
\right)^\alpha\, \dd r
\end{equation}
where the limits from the upper/lower half of the complex $p$-plane are 
identified by the phases $\arg(p) = \pm \, \upi$.
On $[\Omega,\infty[$ the function $\left(1 + \e^{-\ii \upi}\,\Omega/r \right) = 1 - \Omega/r$ 
is non-negative and the integrand of \eqref{eq:Gxx} vanishes.
On $[0,\Omega[$ however $1 + \e^{-\ii \upi}\,\Omega/r = (\Omega/r - 1)\,\e^{-\ii \upi}$ and thus
the integrand of the right-hand side of equation~\eqref{eq:Gxx} does not vanish. Thus
\begin{multline*}
K_\alpha(t,\Omega) = \frac{\sin(\alpha \upi)}{\alpha \upi } \int_0^\Omega \e^{-r\, t} 
\int_0^\Omega \e^{-r\, t}\,r^{-\alpha}\, (\Omega - r)^\alpha\, \dd r = \\
\Omega\; \frac{\sin(\alpha \upi)}{\alpha \upi } \int_0^1 \e^{-\Omega\, t\, y} \, y^{-\alpha}\, (1 - y)^\alpha \, \dd y
\end{multline*}
Hence
\begin{equation} \label{eq:yuy}
\int_0^t K_\alpha(t,\Omega) \, \dd t = \frac{\sin(\alpha \upi)}{\alpha\, \upi } 
\int_0^1 \left( 1 - \e^{-\Omega\, t\, y}\right)\,y^{-\alpha - 1}\, (1 - y)^\alpha \, \dd y
\end{equation}
The integral on the right-hand side of equation~\eqref{eq:yuy} converges if 
$-1 < \alpha < 1$, and 
represents a Bernstein function.
Comparison with the integral representation of the confluent hypergeometric function
(\cite{Abramowitz} Sec.~13.2.1) and the relation $_1\mathrm{F}_1(-\alpha,1;0) = 1$ shows that
\begin{equation} \label{eq:yux}
\frac{\sin(\alpha \upi)}{\alpha\, \upi } 
\int_0^1 \left( 1 - \e^{-\Omega\, t\, y}\right)\,y^{-\alpha - 1}\, (1 - y)^\alpha \, \dd y = 
\left[_1\mathrm{F}_1(-\alpha,1;- \Omega\, t) - 1\right]/\alpha
\end{equation}
provided $-1 < \alpha < 0$. Note also that this expression vanishes at 0.
Consequently if $J_0, M_0\geq 0$ and $-1 < \alpha < 0$  then 
\begin{equation} \label{eq:Strickcreep}
J^{(\alpha,\Omega)}(t) := J_0 + \frac{M_0}{\alpha} \, \left[\, _1\mathrm{F}_1(-\alpha,1;-\Omega\, t) - 1\right], \qquad t > 0
\end{equation}
is a creep compliance. This creep compliance was introduced by E. Strick and F. Mainardi 
\citep{Strick82,StrickMainardi82}.
Note that $\lim_{t\rightarrow 0+} J^{(\alpha,\Omega)}(t) = \lim_{p\rightarrow\infty} = J_0$. If $\alpha < 0$ then the 
infinite time limit of creep compliance  
$J_\infty := \lim_{p\rightarrow 0} \left[ p\, \widetilde{J^{(\alpha,\Omega)}}(p)\right] = J_1$,
where $J_1 := J_0 - M_0/\alpha$, is finite
and $J_\infty \geq J_0$. The Laplace transform of the creep compliance is given by
the formula 
\begin{equation} \label{eq:alpha}
p\, \widetilde{J^{(\alpha,\Omega)}}(p) = J_0 + \frac{M_0}{\alpha} \left[ 
\left( 1 + \frac{\Omega}{p}\right)^\alpha - 1\right]
\end{equation} 
and the retardation spectral density can be calculated from equation~\eqref{eq:yux}: 
$$H^{(\alpha,\Omega)}(r) = \frac{\sin(\alpha \upi)}{\alpha \upi }  
M_0\, r^{-\alpha -1}\, (1 - r/\Omega)^\alpha \, \theta(1 - r/\Omega)$$

The case of $\alpha = 0$ will be treated in a similar way.
$$
K_0(t,\Omega) = \frac{1}{2 \upi \ii} \int_\mathcal{B} \e^{p\, t}\,\ln\left( 1 + \frac{\Omega}{p}\right) \, \dd p = 
-\frac{1}{\upi} \int_0^\infty \e^{-r\, t}\, \im \ln\left(1 + \frac{\Omega}{r} \e^{-\ii \upi}\right) \, \dd r
$$
The logarithm in the integrand is real for $r > \Omega$. On $[0,\Omega[$ however it has the imaginary part 
$-\upi$. Hence 
$$K_0(t,\Omega) = \int_0^\Omega \e^{-r \,t} \, \dd r = \frac{1}{t} \left( 1 - \e^{-\Omega\, t}\right) $$
The indefinite integral of $K_0(t,\Omega)$ 
$$\int_0^t K_0(s,\Omega) \, \dd s \equiv \int_0^\Omega \frac{1}{y} \left( 1 - \e^{-y\, t}\right) \, \dd y $$
is thus a Bernstein function. It is  
recognized as the modified exponential integral 
$\mathrm{Ein}(\Omega \, t)$ \citep{Abramowitz}.
We can now define Becker's creep compliance \citep{Becker25}
\begin{equation} \label{eq:Becker}
J^{(0,\Omega)}(t) = J_0 + M_0 \, \mathrm{Ein}(\Omega\, t), \qquad t \geq 0
\end{equation}
where $J_0$ is a non-negative constant.
Applying the limit $\alpha \rightarrow 0$ in \eqref{eq:alpha} yields
\begin{equation}
p\,\widetilde{J^{(0,\Omega)}}(p) = J_0 + M_0 \, \ln\left( 1 + \frac{\Omega}{p}\right) 
\end{equation}
and the retardation spectral density is 
$H^{(0,\Omega)}(r) = M_0 \, \theta(1 - r/\Omega)/r$.

We have thus proved that the left-hand side of equation~\eqref{eq:yux} is defined 
for $-1 < \alpha < 1$ and is obviously an analytic function of $\alpha$. The confluent hypergeometric function 
is however an analytic function 
of the first argument. Equation~\eqref{eq:yux} therefore holds for $-1 < \alpha < 1$
by analytic continuation, with the value at $\alpha = 0$ given by $\mathrm{Ein}(a\, t)$. 

The creep compliances $J^{(\alpha,\Omega)}(t)$ are shown in Figure~\ref{fig:StrickCompl}.
\begin{figure}
\begin{center}
\includegraphics[width=0.75\textwidth]{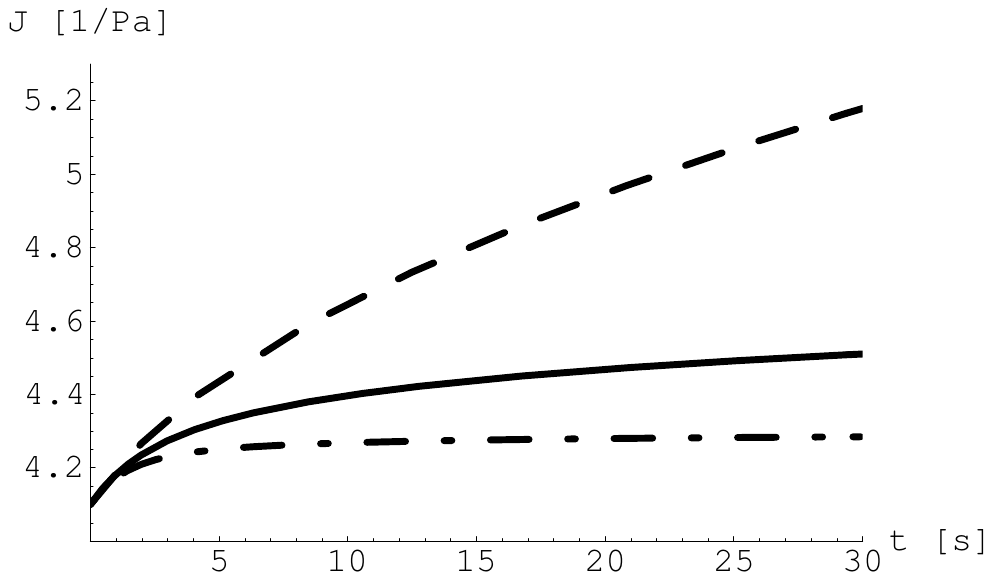}
\end{center}
\caption{Strick-Mainardi creep compliance $J^{(\alpha,1)}$ for $J_0 = 4.1\times 10^{-11}\, \mathrm{Pa}^{-1}$,
$M_0 = 16\times 10^{-11}\,\,\mathrm{Pa}^{-1}/(\upi*50)$ and for $\alpha = -0.5$ (dot-dashed line), $\alpha = 0$ 
(solid line) and $\alpha= 0.5$ (dashed line).} \label{fig:StrickCompl}
\end{figure}

The asymptotic behavior of Strick-Mainardi creep compliance follows from the formulae
\begin{gather}
_1\mathrm{F}_1(a,1;-z) \sim_\infty z^{-a}/\Gamma(1 - a)\\
\mathrm{Ein}(z) \sim_\infty \ln(z) + \gamma + \e^{-z}/z
\end{gather}
(\cite{Abramowitz}, Secs~13.5.1 and 5.1.51)
\begin{equation}
J^{(\alpha,\Omega)}(t) \sim_\infty \begin{cases}
J_0 + M_0\, (\Omega\, t)^\alpha/\alpha  & \alpha > 0\\
J_0 + M_0 \left[ 1 - (\Omega\, t)^\alpha\right]/\vert \alpha\vert  & \alpha < 0\\
J_0 + M_0\, \ln(\Omega\, t) & \alpha = 0
\end{cases}
\end{equation}
Note that the creep compliance for $\alpha < 0$ is bounded and in the remaining cases it is 
unbounded. Hence for $\alpha \geq 0$ the low-frequency limit of creep compliance  
$J_\infty = \infty$ and therefore $G_\infty = 0$. Consequently 
the medium is a viscoelastic solid if $\alpha < 0$ and a viscoelastic fluid if $\alpha \geq 0$.

For $t$ small we can use
the Taylor expansions of the confluent hypergeometric function and the modified exponential
integral: 
\begin{gather}
_1\mathrm{F}_1(a,;z) \sim_0 1 + a \, z /b \\
\mathrm{Ein}(z) \sim_0 z
\end{gather}
(\cite{Abramowitz} Sec.~5.1.53).
Hence the initial rate of creep is approximately linear for $-1 < \alpha <1$.

\section{Attenuation and dispersion in the Strick-Mainardi media.}

We shall now consider the attenuation and dispersion
in materials characterized by the Strick-Mainardi creep compliance $J^{(\alpha,\Omega)}(t)$, where 
$-1 < \alpha < 1$, $\alpha \neq 0$ and $\Omega > 0$:
\begin{equation}
\widetilde{J^{(\alpha,\Omega)}}(p) = J_1 + M_1 \, \left(1 + \frac{\Omega}{p}\right)^\alpha
\end{equation}
where $M_1 := M_0/\alpha$ and $J_1 := J_0 - M_1$. 
The wavenumber function $\kappa(p) = \rho^{1/2}\, p \left[ p \, \widetilde{J^{(\alpha,\Omega)}}(p)\right]^{1/2} 
= p/c_0 + \beta(p)$, where $1/c_0 = (\rho\, J_0)^{1/2}$ if $J_0 > 0$. 

The density of the attenuation-dispersion measure $\nu$ will be calculated from the formula \eqref{eq:hrJ}:
$$h(r) = \frac{\rho^{1/2}}{\upi} \im Z^{1/2}$$
where
$Z := J_1 + M_1\,(1 + \Omega/(r \, \exp(-\ii \upi))^\alpha$.
Note that
$$ \left(1 + \frac{\Omega}{r \, \exp(-\ii \upi)}\right)^\alpha = \begin{cases}
(\Omega/r - 1)^\alpha \, \e^{\ii \upi \alpha}, & r < \Omega\\
(1- \Omega/r)^\alpha, &  r > \Omega
\end{cases}
$$
and $\im Z^{1/2} = \frac{1}{\sqrt{2}} \sqrt{\sqrt{X^2 - Y^2}- X}$,
where $X := \re Z$ and $Y := \im Z$. It follows that $h(r) = 0$ for $r > \Omega$ and
\begin{equation} \label{eq:hStrick}
h(r) = \frac{\rho^{1/2}}{\sqrt{2}\, \upi} \sqrt{\sqrt{X(r)^2 + Y(r)^2} - X(r)}, \qquad r < \Omega
\end{equation}
where $X(r) := J_1 + M_1\, \cos(\alpha\, \upi)\, (\Omega/r - 1)^\alpha$ and
$Y(r) := M_1 \, \sin(\alpha\, \upi)\, (\Omega/r - 1)^\alpha$ for $0 \leq r \leq \Omega$.

The case of $\alpha = 0$ requires some calculi. We note that $h(r)$ is given by equation~\eqref{eq:hStrick}
with $Y(r) = \upi\, M_0$ and  $X(r) = J_0 + M_0\, 
\ln\left(\left\vert\Omega/r - 1\right\vert\right)$, both for $r \leq \Omega$. Hence $h(r)$ vanishes for $r > \Omega$ and
\begin{equation} \label{eq:log}
h(r) \sim_0 \frac{1}{2 \upi\, c_0} \sqrt{\frac{J_0^{\;2} + \upi^2\, M_0^{\;2}}{J_0\, M_0}}\,
\ln^{-1/2}\left(\frac{\Omega}{r} - 1\right), \qquad r < \Omega
\end{equation} 

The attenuation and dispersion can now be determined by substituting \eqref{eq:hStrick} in 
the equations
\begin{gather}
\mathcal{A}(\omega) = \omega^2 \int_{]0,\infty[\;} 
\frac{h(r)}{\omega^2 + r^2} \dd r\label{eq:frAtt}\\
\mathcal{D}(\omega) = \omega \int_{]0,\infty[\;} 
\frac{r\,h(r)}{\omega^2 + r^2} \label{eq:frDisp}
\end{gather}
\citep{HanWM2013}. The attenuation-dispersion 
spectrum (the support of the function $h(r)$) of the materials with the Strick-Mainardi creep 
compliance is bounded. This implies that
the attenuation function tends to a finite value at infinite frequency. In particular, 
if $J_1 = 0$ then
$$h(r) = \frac{\sqrt{\rho \, \vert M_1\vert}}{ \upi} \vert \sin(\alpha\, \upi/2)\vert \, 
\left(\Omega/r - 1\right)^{\alpha/2}, \qquad 0 \leq r \leq \Omega$$
and
$$\int_0^\Omega (\Omega/r - 1)^{\alpha/2} \, \dd r = 
\Omega \int_0^\infty (1 + y)^{-2}\, y^{\alpha/2} \,\dd y
= \frac{\Omega \alpha \upi/2}{\sin(\alpha \upi/2)}$$
hence 
\begin{equation}
\lim_{\omega\rightarrow \infty} \mathcal{A}(\omega) = \int_0^\Omega h(r)\, \dd r = 
 \frac{\vert \alpha\vert \,\Omega\, \sqrt{\rho\, M_1}}{2}
\end{equation}

Attenuation in acoustics is usually expressed in db/m, $A_{\mathrm{dbm}}(\omega) := 
\log_{10}\left(\e^{-\mathcal{A}(\omega)}\right)$,
where $\mathcal{A}(\omega)$ is expressed in m$^{-1}$,  in terms of the quality factor 
$Q(\omega) = \omega/[2 c(\omega)\, \mathcal{A}(\omega)]$ \citep{AkiRichards,Carcione}. 

Figure~\ref{fig:StrickAttCph} shows that the bounded and unbounded Strick-Mainardi  
creep compliances yield very similar dispersion and attenuation.
\begin{figure}
\begin{minipage}{0.48\textwidth}
\includegraphics[width=\textwidth]{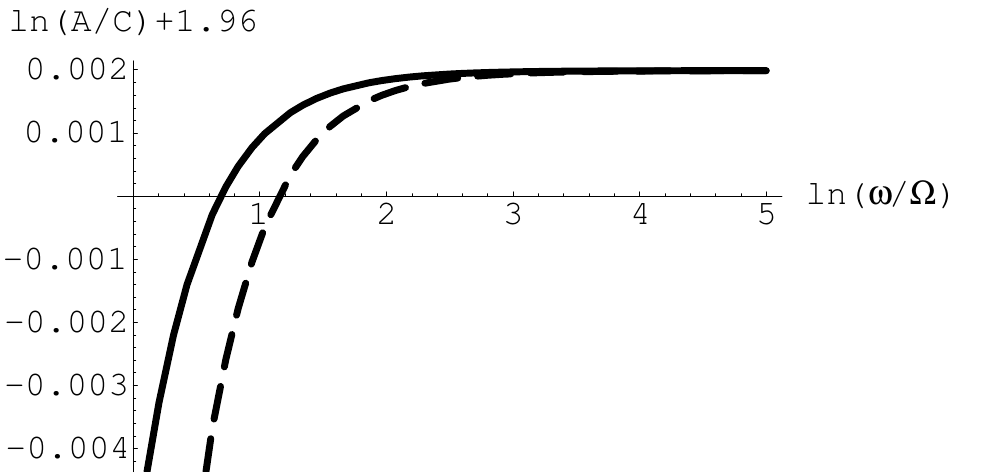}
\begin{center}
{\small (a) Log-log plot of the attenuation functions. }
\end{center}
\end{minipage}
\begin{minipage}{0.48\textwidth}
\includegraphics[width=\textwidth]{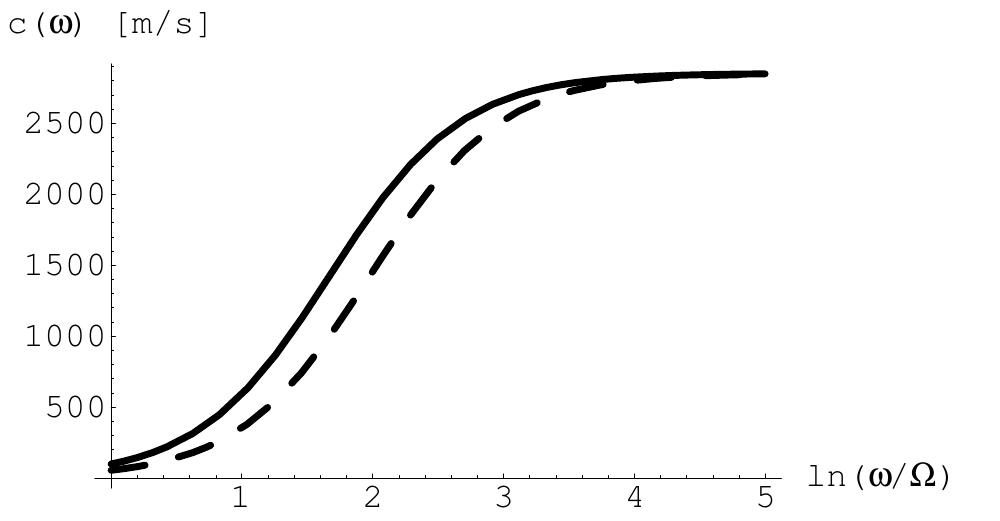}
\begin{center}
{\small (b) Phase speeds.}
\end{center}
\end{minipage}
\caption{Attenuation and phase speed in a medium with the Strick-Mainardi creep compliance 
with $c_0 = 2851 \,\mathrm{m}/\mathrm{s}$,
$J_0 = 4.1\ast 10^{-11}\, \mathrm{Pa}^{-1}$, $M_0 = 0.026$ corresponding to
$Q = 50$. Solid curves: $\alpha=0.3$, dashed curves: $\alpha=-0.3$.} \label{fig:StrickAttCph}
\end{figure}

\section{Wavefronts in Strick-Mainardi models.}

Strick-Mainardi creep models have a bounded creep rate, i.e. $J^{(\alpha,\Omega)\prime}(0) = M_0 \, \Omega$ is finite and 
the jump of Green's function at the wavefront \eqref{eq:g0vsJ} assumes the special form
$\exp(-M_0\, \Omega\, r/(2 c_0\, J_0)$. The values of the Young modulus $1/J_0$ and wavefront velocity $c_0$ or
density $\rho$ are known for many materials and we can only speculate about the creep parameter $M_0$ and the creep 
time scale $2 \upi/\Omega$.
The ratio $M_0/J_0$ controls the rate of gradual creep to instantaneous elastic strain following application of a unit stress.
For $\alpha < 0$ this parameter  controls the 
saturation creep $J_\infty = \lim_{t \rightarrow \infty} J^{(\alpha,\Omega)}(t) = 1/G_\infty$, where $G_\infty$ is the equilibrium elastic modulus.
We recall that $J^{(\alpha,\Omega)}(t)$ tends to infinity as $t \rightarrow \infty$ if $\alpha \geq 0$. 
For a fixed $M_0/J_0$ ratio the logarithmic decay of the wavefront jump is controlled by the wavefront attenuation length scale
$2 \upi \,c_0\,/\Omega$. 

The sign of $\alpha$ determines the long-time asymptotics of the function $g$ and the rate of growth of
Green's function away from the wavefront. Let $0< \alpha < 1$. The asymptotics of 
$h(r)$ for $r \rightarrow 0$
can be easily calculated: $$h(r) \sim_0 \sqrt{\frac{\rho\, M_0\, \Omega^\alpha}{\alpha}} \frac{\sin(\alpha\, \upi/2)}{\upi} 
r^{-\alpha/2}$$
Thus $h(r)$ is regularly varying at 0 and 
\begin{equation}
g(t) \sim_\infty \Gamma(1 - \alpha/2)\sqrt{\frac{\rho\, M_0\, \Omega^\alpha}{\alpha}} 
\frac{\sin(\alpha\, \upi/2)}{\upi} t^{\alpha/2 - 1}
\end{equation}
by the Karamata Abelian Theorem. 
Thus $g(t)$ decreases slower than $t^{-1}$ in this case.

If $-1 < \alpha < 0$, then $(\Omega/r - 1)^\alpha \rightarrow 0$ as $r \rightarrow 0$ and therefore 
$ \lim_{t\,\rightarrow \infty} [t \, g(t)] =  \lim_{r \rightarrow 0} h(r) = 0$. Consequently close to the origin 
the function $g(t)$ decreases faster than $t^{-1}$. 

The case of $\alpha = 0$ has to be considered separately. Equation~\eqref{eq:log} shows that  $h$ is slowly varying at 0. 
Denote the right-hand side of \eqref{eq:log} by $l(\Omega/r)$. It is a function of dimension T/L. We then have
\begin{equation}
g(t) \sim_0 l(\Omega\, t)/t = C /\left[t\, \ln^{1/2}(\Omega \, t)\right] 
\end{equation}
where $C$ is a constant of dimension T/L. Note that $g$ decreases faster than $1/t$. 

The function $g$ can be calculated in closed form in the case of $J_1 = 0$:
\begin{equation}
g(t) = \frac{\alpha\, \sqrt{\rho\, \vert M_1\vert}\, \Omega}{2} \sin(\alpha \upi/2)\, _1F_1(1-\alpha/2,2;-\Omega\, t)
\qquad t \geq 0
\end{equation}

Asymptotic estimates of Green's function in a neighborhood of the
wavefront for $\Omega = M_0 = 1$ are plotted in Figure~\ref{fig:Strickwvf}. Exaggerated values of
material parameters have been chosen for illustrative purposes. For metals Young's modulus
is of order of a few hundreds of GPa. In this case we should assume 
$J_0 \sim 10^{-11} \mathrm{Pa}^{-1}$ and the function $g$ is of order
of $10^{-6}\, \mathrm{m}^{-1}$. With these parameter values the wavefront is hardly 
distinguishable from the simple step function.
Bio-tissues such as liver have however much lower Young's modulus of order of hundreds Pa. 
In this case Green's function exhibits significant variation behind the wavefront. 

\begin{figure}
\begin{minipage}[b]{0.48\linewidth}
\includegraphics[width=\linewidth]{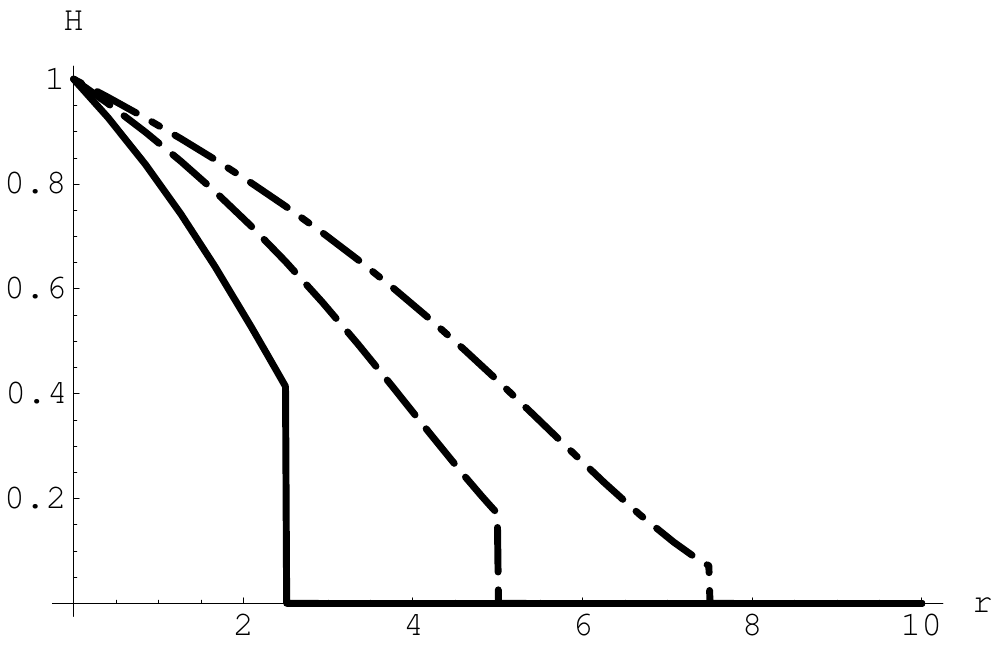}
\begin{center} {\small (a) Evolution of the wavefront profile for $\alpha=0.5$, 
$\Omega = M_0 = 1$, $c_0 = 1 \mathrm{km}/\mathrm{s}$}.
\end{center}
\end{minipage}
\begin{minipage}[b]{0.48\linewidth}
\includegraphics[width=\linewidth]{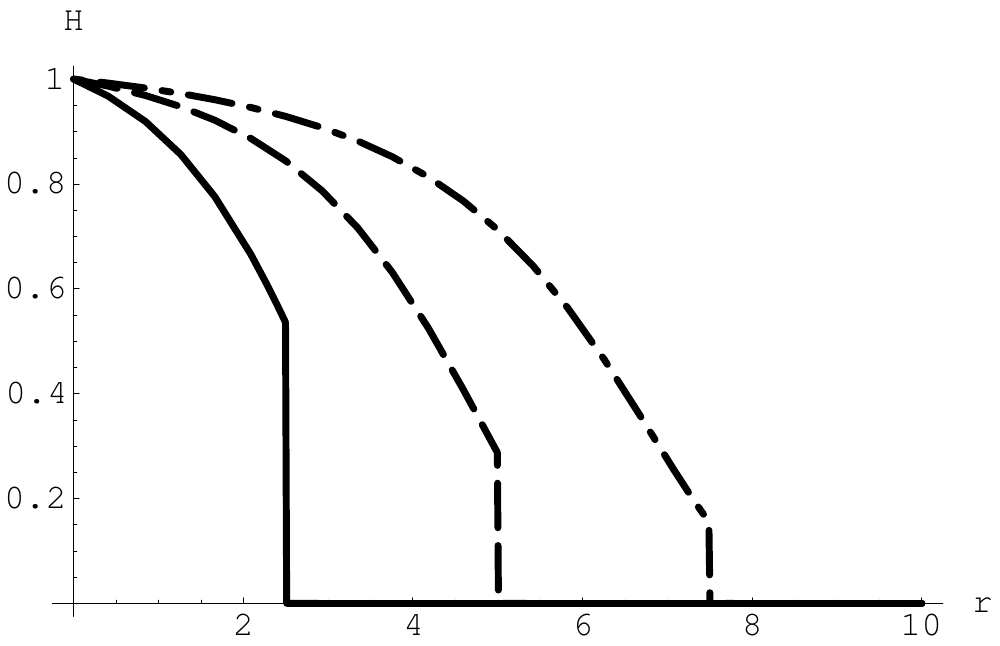}
\begin{center} {\small (b) Evolution of the wavefront profile for $\alpha=-0.5$, 
$\Omega = M_0 = 1$, $c_0 = 1 \mathrm{km}/\mathrm{s}$}.
\end{center}
\end{minipage}
\begin{center}
\begin{minipage}[b]{0.70\linewidth}
\includegraphics[width=\linewidth]{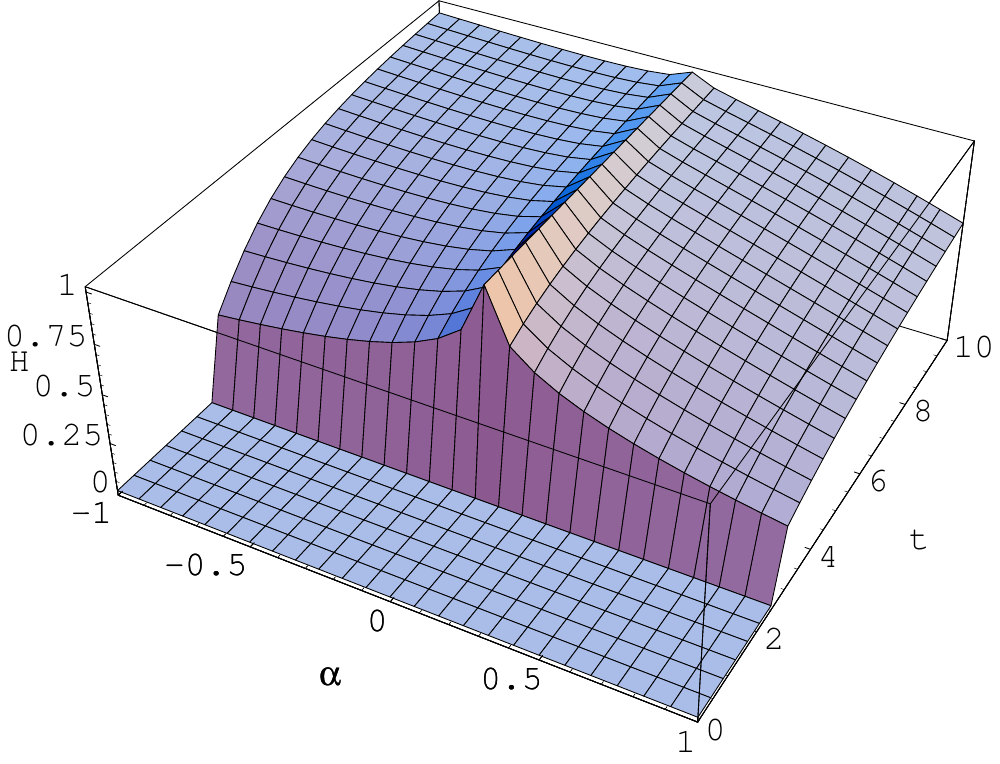}
\begin{center} {\small (c) Dependence of the wavefront signal on $\alpha$ at $r = 5\, \mathrm{km}$.}
\end{center}
\end{minipage}
\end{center}
\caption{Green's function of Strick's creep compliance model near the wavefront.}
\label{fig:Strickwvf}
\end{figure} 

\section{Jeffreys-Lomnitz creep compliance, attenuation and wavefronts.}

The Lomnitz logarithmic law was suggested in the context of the constant $Q$ hypothesis. 
The Jeffreys-Lomnitz-Strick creep compliance is defined by the equation
\begin{equation}
J_{\alpha,\Omega}(t) = J_0 + \begin{cases} 
J_0 + M_0 \frac{(1 + \Omega\, t)^\alpha - 1}{\alpha}, & \alpha \neq 0\\
J_0 + M_0 \, \ln(1 + \Omega\, t), & \alpha = 0
\end{cases}
\end{equation}
for $\alpha \leq 1$, $J_0, M_0, \Omega \geq 0$.
The logarithmic law ($\alpha = 0$) is due to \citet{Lomnitz57,Lomnitz62}, the extension to $\alpha > 0$ was
made by \citet{Jeffreys67} and the extension to negative values of $\alpha$ is due to 
\citet{StrickMainardi82}.  Strick and Mainardi also compared the Jeffreys-Lomnitz-Strick law with Becker's creep compliance, focusing however on the
values of $Q$ predicted by these theories. More recently, the Jeffreys-Lomnitz-Strick law and the associated material response functions were examined by \citet{MainardiSpada2012a}.

The retardation spectral density $H_{\alpha,\Omega}(r)$ of the Jeffreys-Lomnitz-Strick media can be calculated
using the identity \citep{MainardiSpada2012} 
$$\frac{1}{\Gamma(1-\beta)} \int_0^\infty \e^{-r\, t}\, \e^{-r} \, r^{-\beta}\, \dd r = (1 + t)^{\beta-1}$$
for $\beta < 1$. This identity is easily proved by substituting $s = (1 + t)\, r$. 
It follows that
\begin{equation} \label{eq:JFLret}
J_{\alpha,\Omega}(t) = J_0 + \frac{M_0}{\Omega \, \Gamma(-\alpha)} \int_0^\infty \left( 1 - \e^{-r\, t}\right)\, 
\e^{-r/\Omega}\, (r/\Omega)^{-\alpha}\, \dd r
\end{equation}
for $\alpha \neq 0$. For $\alpha = 0$ we note that equation~\eqref{eq:JFLret} follows from the identity 
$$F(x) := \int_0^\infty \left(1 - \e^{-x\, y}\right) \, \e^{-y}\, y^{-1}\, \dd y = \ln(1 + x)$$
Indeed, 
$F^\prime(x) = \int_0^\infty \e^{-x y - y} \, \dd y = 1/(1 + x)$.
Hence the Jeffreys-Lomnitz-Strick retardation spectral density is given by the formula 
\begin{equation}
H_{\alpha,\Omega}(r) = \frac{M_0}{\Omega\, \Gamma(-\alpha)} \e^{-r/\Omega} \, (r/\Omega)^{-\alpha} 
\end{equation}
Note that $\int_0^\infty H_{\alpha,\Omega}(r)\, \dd r < \infty$. 

The attenuation spectral density is more difficult to calculate. 
The Laplace transform of the Jeffreys-Lomnitz-Strick creep compliance can be expressed 
in terms of the exponential integral \citep[Chap.~5]{Abramowitz}
$$\mathrm{E}_\alpha(q) := \int_1^\infty \e^{-q\, r} \, r^{-\alpha} \, \dd r$$
(do not confuse this notation with the Mittag-Leffler function)
by the formula
\begin{equation}
\widetilde{J_{\alpha,\Omega}}(p) = \left\{ J_0 + M_0\, \left(p\, \e^{p/\Omega}\, \mathrm{E}_{-\alpha}(p/\Omega)/\Omega - 1\right)/\alpha \right\}/p, \qquad \alpha \neq 0
\end{equation}
For $\alpha = 0$ note that
$$\int_0^\infty \e^{-p t}\, \ln(1 + \Omega\, t) \, \dd t = \frac{\e^{p/\Omega}}{\Omega}
\int_1^\infty \e^{-p y/\Omega} \,\ln(y)\, \dd y = \frac{\e^{p/\Omega}}{p} \int_1^\infty \e^{-p y/\Omega}/y\, \dd y$$
Hence
\begin{equation}
\widetilde{J_{\alpha,\Omega}}(p) = \left[J_0 + M_0 \, \e^{p/\Omega}\, \mathrm{E}_1(p/\Omega)\right]/p
\end{equation}
The exponential integral  
has a branching cut along the entire negative axis. We thus do not expect the attenuation 
spectrum to be bounded, but we shall show that the attenuation measure has finite total mass
and therefore the attenuation function is bounded. This implies 
that $g(0+) < \infty$ and shock wave discontinuities propagate at the wavefronts.
Consequently Jeffreys-Lomnitz-Strick media support shock waves. 

The asymptotic formula 
$$\mathrm{E}_{-\alpha}(z) \sim_\infty \left(\e^{-z}/z\right) \, \left[ 1 + \alpha/z + \OO\left[z^{-2}\right]
\right]$$
(\citet{Abramowitz}, Sec.~5.1.51) implies that
$$\lim_{p\rightarrow \infty}  \left\{ p\, \left[\left(\rho\, p \,
\widetilde{J_{\alpha,\Omega}}(p)\right)^{1/2} - (\rho\, J_0)^{1/2}\right]\right\} = N$$
where 
$ N = M_0/[2 c_0\, J_0] < \infty$ 
for both $\alpha \neq 0$ and $\alpha = 0$. 
Hence, by Theorem~\ref{thm:xx} 
 the attenuation measure $\nu$ has finite total mass. It follows that 
$\lim_{\omega\rightarrow\infty}
\mathcal{A}(\omega) = \int_{]0,\infty[} \nu(\dd r) = N$, hence the attenuation 
function is bounded.

In the case at hand the attenuation function and the phase speed can be calculated using 
the equations
$\mathcal{A}(\omega) = \omega\, \im\left[-\ii \omega \, 
\widetilde{J_{\alpha,\Omega}}(-\ii \omega)\right]^{1/2}$,
$\mathcal{D}(\omega) = \omega\, \re\left[-\ii \omega \, \widetilde{J_{\alpha,\Omega}}(-\ii \omega)\right]^{1/2}$ 
and $1/c(\omega) = 1/c_0 + \mathcal{D}(\omega)/\omega$, where $c_0 = (\rho\, J_0)^{-1/2}$.  
The results for selected parameters are shown in Figure~\ref{fig:JLS}.
\begin{figure}
\begin{minipage}{0.48 \textwidth}
\begin{center}
\includegraphics[width=\textwidth]{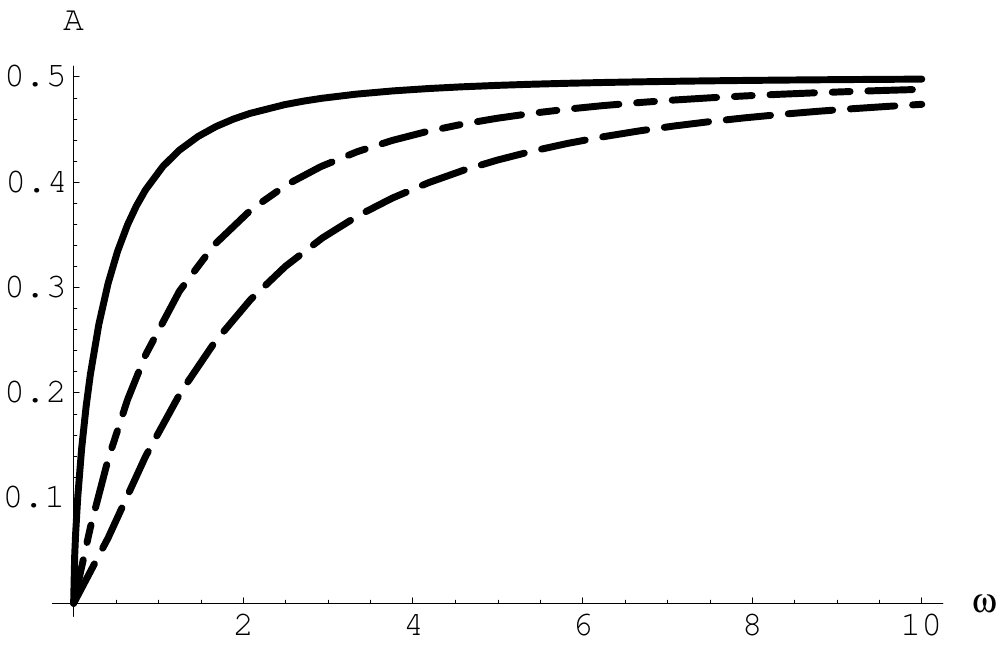}
 {\small (a) Attenuation function.} 
\end{center}
\end{minipage}
\begin{minipage}{0.48\textwidth}
\begin{center}
\includegraphics[width=\textwidth]{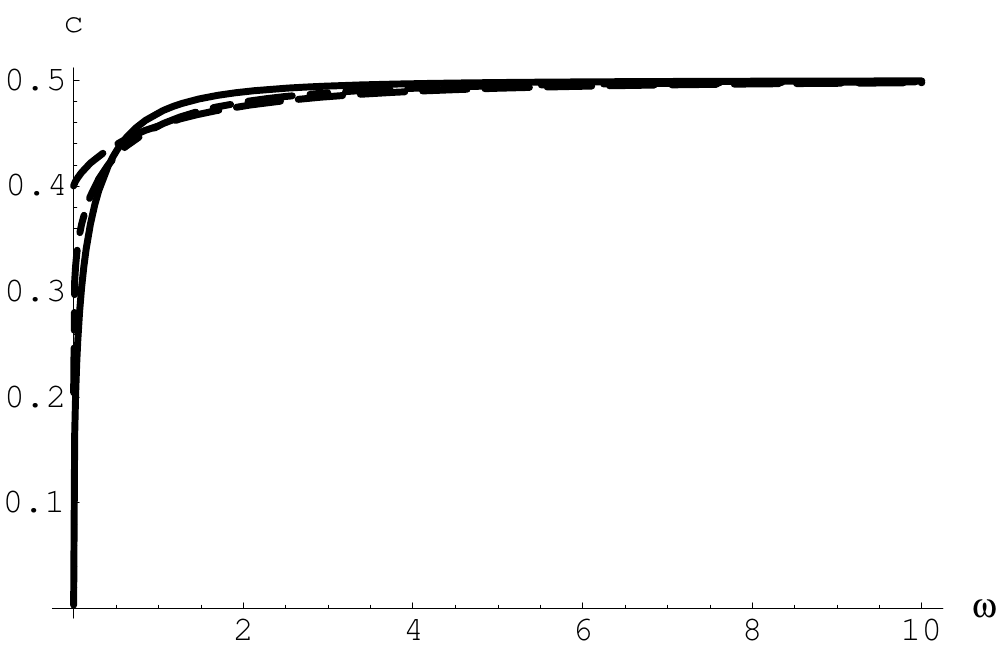}
 {\small (b) Phase speed.} 
\end{center}
\end{minipage}
\caption{The attenuation and phase speed for the Jeffreys-Lomnitz-Strick media, for $J_0 = M_0 = \Omega = 1$,
solid line: $\alpha = 0.8$; dashed line: $\alpha = 0$; dot-dashed line: $\alpha = -0.8$.}
\label{fig:JLS}
\end{figure}

\section{A viscoelastic medium exhibiting the pedestal effect.}

The Jeffreys creep compliance with $0 < \alpha < 1$ can be recast in the form
\begin{equation}
J_{\alpha,\tau}(t) = J_0 + J_2\, (\tau + t)^\alpha  
\end{equation}
where $\tau = 1/\Omega$ and $J_2 = M_0\, \Omega^\alpha/\alpha$. $J_{\alpha,0}$ is a special case 
of the Andrade creep compliance 
\begin{equation}
J_\alpha(t) = J_0 + J_1\, t + J_2\, t^\alpha 
\end{equation}
where $J_0, J_1, J_2 \geq 0$ and $0 < \alpha \leq 1$. The second term is known as linear creep; 
 it dominates at long observation times. Pure linear creep is equivalent to Newtonian viscosity.
The third term dominates for shorter observation times and is attributed to dislocation motion.
Andrade creep was originally observed in metals 
\citep{Andrade1,Andrade2,Cottrell96,Nabarro97,MiguelAl2002} with $\alpha = 1/3$, 
but it was subsequently found in other 
materials, including rocks \citep{Lockner93,MurrellChakravarthy73,GribbAl98}. 
$J_\alpha$ is clearly a Bernstein function. 

Even though the Andrade creep compliance has been obtained by 
a limiting process from the Jeffreys creep compliance it is radically different 
from the latter because $J_\alpha^\prime(t)$ tends to infinity for $t \rightarrow 0+$.
It is shown in \citep{HanUno} that viscoelastic media with this property do not support
discontinuity waves. 

Note that 
$p\, \widetilde{J_\alpha}(p) = J_0 + J_1/p + J_2\, p^{-\alpha}$ and
$\kappa(p) = p\, \left[ 1 + J_1/(J_0\, p) + J_2/J_0 \, p^{-\alpha}\right]^{1/2}/c_0$,
where $c_0$ is given by equation~\eqref{eq:c0}. Hence
$$\lim_{p\rightarrow\infty}\, \left[ \kappa(p) - p/c_0\right] = \frac{1}{2 c_0} 
\lim_{p\rightarrow\infty}\, \left[\frac{J_1}{J_0} + \frac{J_2}{J_0} p^{1-\alpha}\right] = 
\infty$$
Hence the attenuation measure $\nu$ has infinite total mass and the attenuation
spectrum is the entire positive semi-axis. 

The attenuation function can be explicitly calculated 
\begin{multline}
\mathcal{A}(\omega) = \re \kappa(-\ii \omega)  = 
\frac{\omega}{c_0} \im \left[ 1 + \frac{J_1}{-\ii \omega\, J_0} +
\frac{J_2}{J_0} (-\ii \omega)^{-\alpha}\right] = \\
\frac{\omega}{c_0\, \sqrt{2}}\Big\{
\sqrt{1 + J_1^{\;2}/(J_0^{\;2}\, \omega^2) + J_2^{\;2}/J_0^{\;2}\, \omega^{-2 \alpha}
+ J_2/J_0 \, \left[\sin{\upi \alpha/2) + \cos(\upi\, \alpha/2)}\right]\, \omega^{-\alpha}}
\\ -1 - J_2/J_0 \, \cos(\upi \alpha/2)\, \omega^{-\alpha}
\Big\}
\end{multline}
In the high-frequency range
\begin{equation}
\mathcal{A}(\omega) \sim_\infty \sin^{1/2}(\upi\, \alpha/2)\, \left(\frac{J_2}{J_0}\right)^{1/2}
\frac{\omega^\gamma}{\sqrt{2}}  
\end{equation}
where $\gamma := 1 - \alpha/2$ satisfies the inequalities $1/2 < \gamma < 1$. 
It follows from the theory developed by \cite{HanJCA} that Green's functions for the 
Andrade viscoelastic media
are infinitely smooth at the wavefronts. This in turn implies that acoustic pulses 
follow the wavefront with a delay and are preceded by a pedestal in Strick's terminology
\citep{Strick1:ConstQ}. The effective travel time of a seismic signal is thus greater
than the wavefront travel time, which is directly linked to the wavefront speed. In seismic inversion the effective travel time is determined \citep{FastTrack}. If it is believed that 
the attenuation function increases at a rate higher than logarithmic,  as is the case in
Andrade viscoelastic media, then the standard methods of seismic inversion misposition
the scatterers. This error was pointed out by \citep{FastTrack} and in a different context by \citet{Strick3:ConstQ}.

The Andrade model was conceived as
a fit to creep data rather than wave propagation. It is nevertheless an instructive example
of the possibility of an entirely different wave propagation pattern, with important 
consequences for the identification of travel times and location of scatterers in 
seismic applications.

Attenuation and dispersion in Andrade viscoelastic media was recently studied by 
semi-numerical methods by \cite{BenJaziaLombardBellis2013}. The discrete approximation of
the Andrade creep compliance applied in this paper does not reflect the unboundeness 
of the attenuation spectrum and of the attenuation function. Consequently it does not
account for the wavefront smoothness.

\section{Concluding remarks.}

Strick-Mainardi and Jeffreys-Lomnitz-Strick viscoelastic models provide the only known examples of a closed form creep 
compliance consistent
with the propagation of shock waves. The former models are characterized by bounded attenuation and retardation spectra while
the latter have integrable attenuation and retardation spectral densities. The Strick-Mainardi retardation and 
attenuation-dispersion spectral measures are given by elementary functions and the function $g$ is easy to analyze.
The Jeffreys-Lomnitz-Strick models are not amenable to such a detailed analysis but 
numerical analysis
shows that they they are qualitatively fairly similar. The similarity is due to the fact 
that both classes of models defined in terms of 
the power function, which is invariant with respect to the Carson-Laplace transform up
to a numerical factor.
In the context of the $Q$ factor such 
striking similarities were discovered in \cite{StrickMainardi82}.

Short-time creep and its singularity at 0 affects Green's functions at the wavefront.
The exact time dependence of the wave field at the wavefront is however represented by 
the function $g(t)$, which is indirectly related to the creep rate function. 
The wavefront behavior of Green's functions provides a constraint 
on the creep rate at short times.  

\section{Acknowledgment.}

The Author is indebted to Francesco Mainardi for precious bibliographic information.

\end{document}